\documentclass[onecolumn,nofootinbib,thightenlines,notitlepage,tightenlines,longbibliography,11pt]{revtex4-1} 
\newcommand{\pagenumbaa}{1}
\usepackage{graphicx}
\usepackage{amssymb}
\usepackage{amstext}
\usepackage{enumerate}

\usepackage{tikz}
\usetikzlibrary{chains}
\usetikzlibrary{fit}
\usepackage{pgflibraryarrows}		
\usepackage{pgflibrarysnakes}		
\usepackage{xcolor}
\usepackage{epsfig}
\usetikzlibrary{shapes.symbols,patterns} 

\usepackage{amsthm}
\usepackage{amsmath}
\usepackage{amsfonts}
\usepackage{amssymb,amstext}
\usepackage{bbm} 
\usepackage{dsfont}

\usepackage[colorlinks=true,urlcolor=blue, hyperindex,breaklinks=true] {hyperref}
\usepackage{todonotes}

\usepackage{footnote}

\theoremstyle{plain}
\newtheorem{mythm}{Theorem}
\newtheorem{myprop}[mythm]{Proposition}
\newtheorem{mycor}[mythm]{Corollary}

\theoremstyle{definition}
\newtheorem{mydef}{Definition}

\newcommand{\prlsection}[1]{\emph{#1}.---}

\newcommand{\Hmin}[1]{H_{\mathsf{min}}\!\left({#1}\right)} 
\newcommand{\Hcmin}[2]{H_{\mathsf{min}}\!\left({#1}\!\left|{#2}\right.\right)} 
\newcommand{\cY}{\mathsf{c}_Y} 
\newcommand{\cZ}{\mathsf{c}_Z} 

\def\BSC{\text{BSC}}
\def\BEC{\text{BEC}}
\def\Hb{\ensuremath{H_{b}}}

\def\setC{\mathsf{c}}

\newcommand{\Hh}[1]{H\!\left({#1}\right)} 
\newcommand{\Hc}[2]{H\!\left({#1}\!\left|{#2}\right.\right)} 
\newcommand{\I}[2]{I\!\left({#1};{#2}\right)} 
\newcommand{\Ic}[3]{I\!\left({#1};{#2}\!\left|{#3} \right. \right)} 

\newcommand{\Perr}[2]{{p_{\rm err}}\!\left(\left.#1\right|#2\right)} 
\newcommand{\Pguess}[2]{{p_{\rm guess}}\!\left(\left.#1\right|#2\right)} 

\newcommand{\W}{\mathsf{W}} 
\newcommand{\V}{\mathsf{V}} 

\newcommand{\markovDavid}{\small{\mbox{$-\hspace{-1.3mm} \circ \hspace{-1.3mm}-$}}}

\usepackage{makecell}

\begin{document}

\title{Universal Polar Codes for More Capable and Less Noisy Channels and Sources}

 \author{David Sutter}
 \email[]{suttedav@phys.ethz.ch}
 
 \author{Joseph M.\ Renes}
 \email[]{renes@phys.ethz.ch}
 \affiliation{Institute for Theoretical Physics, ETH Zurich, Switzerland}

%




\begin{abstract}
We prove two results on the universality of polar codes for source coding and channel communication.  First, we show that for any polar code built for a source $P_{X,Z}$ there exists a slightly modified polar code---having the same rate, the same encoding and decoding complexity and the same error rate---that is universal for every source $P_{X,Y}$ when using successive cancellation decoding, at least when the channel $P_{Y|X}$ is \emph{more capable} than $P_{Z|X}$ and $P_X$ is such that it maximizes $\I{X}{Y}-\I{X}{Z}$ for the given channels $P_{Y|X}$ and $P_{Z|X}$. This result extends to channel coding for discrete memoryless channels. Second, we prove that polar codes using successive cancellation decoding are universal for \emph{less noisy} discrete memoryless channels.
\end{abstract}

 \maketitle 

 \setcounter{page}{\pagenumbaa}  
 \thispagestyle{plain}
\vspace{6mm}

In source and channel coding scenarios, it is usually assumed that the statistics of the source or channel  is known precisely. However, this cannot be always guaranteed in practice. Hence, it is of great interest---both practically and theoretically---to have a code that remains reliable and efficient for a large class of different sources or channels. Such robust codes are called \emph{universal}. There are many different notions of universality and many different problems of how to communicate reliably when the statistics of the source or channel cannot be ensured to be known precisely. An excellent overview can be found in \cite{lapidoth98}.


\c Sa\c so\u glu showed in \cite[Chapter 7]{sasoglu_phd} that any capacity achieving code built for some discrete memoryless channel (DMC) $\W$ will also be reliable over any DMC $\V$ that is more capable than $\W$ when using an optimal decoder. With a similar proof technique it can be shown that the same is true for the task of source coding with side information. Assume that we have two sources $P_{X,Y}$ and $P_{X,Z}$ where $Y$ and $Z$ denote the side information such that $P_{Y|X}$ is more capable than $P_{Z|X}$, then any optimal code for $P_{X,Z}$ will also be reliable for $P_{X,Y}$ using an optimal decoder. In the first part of this article we will prove a stronger statement for the use of polar codes by showing that for any polar code constructed for some source $P_{X,Z}$ there exists a slightly modified polar code---having the same rate, the same efficient encoding and decoding and the same error rate---that can be used reliably for any source $P_{X,Y}$ using successive cancellation (SC) decoding, where $P_{Y|X}$ is more capable than $P_{Z|X}$ and $P_X$ is such that it maximizes $\I{X}{Y}-\I{X}{Z}$ for the given channels $P_{Y|X}$ and $P_{Z|X}$. This result carries over to channel coding for DMCs, implying that for any polar code for some symmetric DMC $\W$ there exists a slightly modified polar code having the same rate that can be used reliably for an arbitrary symmetric DMC $\V$ that is more capable than $\W$ using SC decoding, whenever the input distribution $P_X$ is chosen such that it maximizes $\I{X}{Y}-\I{X}{Z}$ for the given channels $P_{Y|X}$ and $P_{Z|X}$.

The second part of this article is concerned with the universality of polar codes under SC decoding. A polar code for some DMC $\W:\mathcal{X}\to \mathcal{Y}$ can be defined via a so-called \emph{high-entropy set} $\mathcal{D}_{\epsilon}^n(X|Y)$ as defined in \eqref{eq:DsetY}. Note that it is in general hard to compute this set. The only DMC for which the code construction is known to be simple is the binary erasure channel ($\BEC$) \cite{arikan09,sasogluPC}. 
Recently, an algorithm has been introduced that approximately can do the code construction for arbitrary DMCs with a fixed input alphabet, having a linear running time \cite{talandvardy10,tal12}. Very recently a slightly different algorithm for constructing polar codes has been introduced \cite{guruswami13}. Although these algorithms seem to solve the problem of the code construction for polar codes to a large extent, it is still interesting to know whether for example the code for some DMC could be also used for different DMCs. One reason is that it is not clear how good a polar code constructed by the algorithm \cite{talandvardy10,tal12} approximates the capacity. A precise analysis of these algorithms has been carried out in \cite{pedarsani11}. The more recent algorithm introduced in \cite{guruswami13}, which does guarantee an approximation to the capacity, has been shown to have a polynomial running time, which might not be useful in practice.\footnote{Since the exact running time has not been derived, it could be that the polynomial is very large making the algorithm impractical.} Another reason concerns universality. Assuming that one does not know the channel characteristics precisely, it would be useful to know how sensitive the code construction is with respect to different channel statistics.
Ar{\i}kan showed that a polar code constructed for some DMC $\W$ can be used for any other DMC $\V$ of which $\W$ is a degraded version, using SC decoding \cite{arikan09}. In this article we generalize this statement by proving that a polar code built for some DMC $\W$ can be used for any other DMC that is less noisy than $\W$, again using SC decoding. 

\prlsection{Previous work}
In \cite{hassani09}, it has been shown that polar codes are not universal for symmetric DMCs under SC decoding by proving that the compound capacity of a binary symmetric channel and a binary erasure channel having both capacity $\tfrac{1}{2}$, is strictly smaller than $\tfrac{1}{2}$ when using SC decoding. In some recent work \cite{hassani13}, two modifications of the standard polar coding scheme are introduced that achieve the compound capacity with low complexity for any  binary input symmetric output memoryless channel, but at the price of a considerably larger blocklength. Another recent work \cite{sasogluLele13} introduces a different modification of the standard polar coding scheme that achieves the compound capacity, again at the cost of a larger blocklength.
Note that the modified polar code we introduce in this article is totally different from these two recent results and does inherit the low-complexity and error decay from the standard polar coding scheme. A very recent work \cite{alsan13} is concerned with the universality of polar codes when removing the assumption that the decoder has knowledge about the precise channel statistics. Note that in this article we assume that the decoder always knows the precise channel statistics. The use of polar codes in more general broadcast scenarios has been studied in \cite{goela13}.

\prlsection{Structure}
The rest of this article is structured as follows. After introducing some notation and preliminary results in Section~\ref{sec:prelimiaries}, we prove in Section~\ref{sec:universalityMC} that \c Sa\c so\u glu's result---that any good code built for some DMC $\W$ will also perform well, using an optimal decoder, for any DMC that is more capable than $\W$---carries over to the setup of source coding with side information. In addition, we prove a stronger statement for polar codes. Section~\ref{sec:universalityMCchannels} shows that the result from Section~\ref{sec:universalityMC} remains valid for channel coding when considering DMCs. Finallly, in Section~\ref{sec:UniversalityPC} we show that any polar code built for some DMC $\W$ can be used reliably for any DMC that is less noisy than $\W$, using SC decoding.

\section{Preliminaries} \label{sec:prelimiaries}
We start by introducing some standard notation. Let $[k]=\left \lbrace 1,\ldots,k \right \rbrace$ for $k\in \mathbb{Z}^+$. For $x \in \mathbb{Z}_2^k$ and $\mathcal{I}\subseteq [k] $ we have $x[\mathcal{I}]=[x_i:i\in \mathcal{I}]$, $x^i=[x_1,\ldots,x_i]$ and $x_j^i=[x_j,\ldots,x_i]$ for $j\leq i$. For two sets $\mathcal{A},\mathcal{B}\subseteq [n]$  we write $\mathcal{A}$ $\scriptsize{\overset{\boldsymbol{\cdot}}{\subseteq}}$ $\mathcal{B}$ meaning that $\mathcal{A}$ is essentially contained in $\mathcal{B}$ or more precisely $|\mathcal{A}\backslash \mathcal{B}|=o(n)$. We write $\mathcal{A}\overset{\boldsymbol{\cdot}}{=}\mathcal{B}$ if $|\mathcal{A}\backslash \mathcal{B}| =o(n) = |\mathcal{B}\backslash \mathcal{A}|$. The complement of $\mathcal{A}$ in $[n]$ is denoted by $\mathcal{A}^{\setC}:=[n]\backslash \mathcal{A}$. All logarithms in this article are with respect to the basis $2$. For $\alpha \in [0,1]$, $\Hb(\alpha):=-\alpha \log \alpha - (1-\alpha) \log(1-\alpha) $ denotes the binary entropy function. We denote the min-entropy of some random variable $X$ by $\Hmin{X}:=-\log \max_{x\in\mathcal{X}} P_X(x)$ and the conditional min-entropy of $X$ given $Y$ by $\Hcmin{X}{Y}:= - \log \max_{x\in \mathcal{X}, y\in\mathcal{Y}}P_{X|Y=y}(x)$. Note that by definition it follows that $\Hmin{X}\leq \Hh{X}$ and $\Hcmin{X}{Y} \leq \Hc{X}{Y}$, where $\Hh{X}$ and $\Hc{X}{Y}$ denote the standard Shannon entropies. For two random variables $X$ and $Y$ we denote by $\Pguess{X}{Y}$ the probability of correctly guessing $X$ given $Y$ when using an optimal strategy. The error probability of $X$ given $Y$ can then be defined as $\Perr{X}{Y}:=1-\Pguess{X}{Y}$.

Let $X^n$ be a vector whose entries are i.i.d.\ Bernoulli($p$) distributed for $p\in[0,1]$ and  $n=2^k$ where $k\in \mathbb{N}$. Then, define $U^n = G_n X^n$, where $G_n$ denotes the polarization (or polar) transform which can be represented by the matrix
\begin{equation}
G_n := \begin{pmatrix}
1 & 1\\
0 & 1
\end{pmatrix}^{\!\! \otimes  \log n}, \label{eq:polarTrafo}
\end{equation}
where $A^{\otimes k}$ denotes the $k$th Kronecker power of an arbitrary matrix $A$. Furthermore, let $Y^n = \W^n X^n$, where $\W^n$ denotes $n$ independent uses of a DMC $\W:\mathcal{X}\to \mathcal{Y}$ and let $Z^n=\V^n X^n$, where $\V:\mathcal{X} \to \mathcal{Z}$ denotes another DMC. For any $\epsilon \in (0,1)$ we consider the two low-entropy sets
\begin{align}
\mathcal{D}_{\epsilon}^n(X|Y)&:= \left \lbrace i \in[n]: \Hc{U_i}{U^{i-1},Y^n}\leq \epsilon \right \rbrace \quad \textnormal{and}\label{eq:DsetY}\\
\mathcal{D}_{\epsilon}^n(X|Z)&:= \left \lbrace i \in[n]: \Hc{U_i}{U^{i-1},Z^n}\leq \epsilon \right \rbrace,\label{eq:DsetZ}
\end{align}
which define a polar code for $\W$ respectively $\V$ that is reliable using SC decoding.

In this article we will make use of different relations between DMCs which we introduce next.
Let $\mathcal{C}$ denote the class of all binary input symmetric output (BISO) channels. The set $\mathcal{C}_C$ denotes the class of all BISO channels with fixed capacity $C$. Note that the binary symmetric channel ($\BSC$) and the binary erasure channel ($\BEC$) belong to $\mathcal{C}$. Furthermore, $\BSC(\alpha)$ with $C=1-\Hb(\alpha)$ and $\BEC(\beta)$ with $C=1-\beta$ belong both to the class $\mathcal{C}_C$. We denote by $\BSC_C$ the binary symmetric channel having capacity $C$ and by $\BEC_C$ the binary erasure channel having capacity $C$. 


\begin{mydef}
Let $\W:\mathcal{X} \to \mathcal{Y}$ and $\V:\mathcal{X}\to \mathcal{Z}$ be two DMCs then
\begin{itemize}
\item $\W$ is \emph{more capable} than $\V$ (denoted by $\W \gg \V$) if $\I{X}{Y}\geq \I{X}{Z}$ for all distributions $P_X $.
\item $\W$ is \emph{less noisy} than $\V$ (denoted by $\W \succeq \V$) if $\I{U}{Y}\geq \I{U}{Z}$ for all distributions $P_{U,X}$ where $U$ has finite support and $U\markovDavid X \markovDavid (Y,Z)$ form a Markov chain.
\item $\V$ is said to be a \emph{degraded} version of $\W$ if there exists a channel $\mathsf{T}:\mathcal{Y}\to \mathcal{Z}$ such that $\V(z|x)=\sum_{y \in \mathcal{Y}} \W(y|x)\mathsf{T}(z|y)$ for all $x \in \mathcal{X}$, $z\in \mathcal{Z}$.
\end{itemize}
\end{mydef}
\begin{figure}[!htb]
\centering
\def \x{1.0}
\def \s{0.2}

\begin{tikzpicture}[scale=1,auto, node distance=1cm,>=latex']
	
 \draw[dotted] (0,0) ellipse (0.55cm and 0.35cm); 
 \draw (0,0) ellipse (1cm and 0.6cm); 
 \draw[dashed] (-\s,0) ellipse (1.5cm and 0.8cm); 
  
    \node at (3.15,-0.5) {less noisy};    
    \node at (-3.3,0.5) {more capable};
    \node at (-2.95,-0.5) {degraded};         
 \draw[->] (2.35,-0.5) -- (0.75,0);   
 \draw[->] (-2.2,0.45) -- (-1.5,0);   
 \draw[->] (-2.2,-0.45) -- (0,0);                   
\end{tikzpicture}
\caption{\small The relations between the three classes \emph{degraded}, \emph{less noisy} and \emph{more capable}. The proofs can be found in \cite{nair09,korner75}.}
\label{fig:relations}
\end{figure}
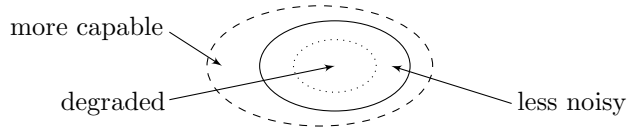

\begin{myprop}[{\cite{csiszar78},\cite[Page 448]{csiszarkorner81}}] \label{prop:csiszarTrick}
Let $\W:\mathcal{X}\to \mathcal{Y}$ and $\V:\mathcal{X}\to\mathcal{Z}$ be two DMCs such that $\W \gg \V$. Then $\Ic{X}{Y}{U} \geq \Ic{X}{Z}{U}$ for all $P_{U,X}$ where $U \markovDavid X \markovDavid (Y,Z)$.
\end{myprop}

\begin{myprop}[\cite{geng10}] \label{prop:extremes}
Let $F \in \mathcal{C}_C$, then $F \gg \BSC_C $ and $\BEC_C \gg F$.
\end{myprop}

\begin{myprop}[\cite{geng10}] \label{prop:moreCapableComparable}
Let $F$ and $G$ be two arbitrary channels in $\mathcal{C}_C$ whose output alphabet size is at most $3$, then either $F \gg G$ or $F \ll G$, i.e., $F$ and $G$ are more capable comparable.
\end{myprop}


The more capable as well as the less noisy relation are preserved under products of more capable, respectively less noisy, channels.
\begin{myprop}[\cite{csiszarkorner81}] \label{prop:tensorMC}
Let $\W:\mathcal{X}\to\mathcal{Y}$ and $\V:\mathcal{X}\to \mathcal{Z}$ be two DMCs such that $\W\gg \V$. Then $\W^n \gg \V^n$.
\end{myprop}

\begin{myprop} \label{prop:tensorEssLN}
Let $\W:\mathcal{X}\to\mathcal{Y}$ and $\V:\mathcal{X}\to \mathcal{Z}$ be two DMCs such that $\W\succeq \V$. Then $\W^n \succeq \V^n$.
\end{myprop}
\begin{proof}
We start with a proof for the case $n=2$, which we then generalize to arbitrary $n\in\mathbb{N}$. Using the chain rule we can write for an arbitrary $U$ satisfying $U \markovDavid X^2 \markovDavid(Y^2,Z^2)$
\begin{align}
&\I{U}{Y_1,Y_2}- \I{U}{Z_1,Z_2} \nonumber \\
&\hspace{5mm}= \Ic{U}{Y_2}{Y_1}-\Ic{U}{Z_1}{Z_2} + \I{U}{Y_1} - \I{U}{Z_2}\\
&\hspace{5mm}= \Ic{U}{Y_2}{Y_1}-\Ic{U}{Z_1}{Z_2} +\I{U}{Y_1}-\I{U}{Z_2}+\I{U}{Y_1,Z_2}-\I{U}{Y_1,Z_2}\\
&\hspace{5mm}= \Ic{U}{Y_1}{Z_2}-\Ic{U}{Z_1}{Z_2} + \Ic{U}{Y_2}{Y_1}-\Ic{U}{Z_2}{Y_1}\\
 &\hspace{5mm}\geq 0, \label{eq:ineqCN}
\end{align}
where the inequality step can be justified as follows. Some given conditioning $Z_2=z_2$ specifies a probability distribution of $X_1$ given $Z_2=z_2$ and $U$ given $Z_2=z_2$. By definition of $\W\succeq \V$ it follows that $\I{U}{Y_1}\geq \I{U}{Z_1}$ for all $P_{U,X_1}$ such that $U \markovDavid X_1 \markovDavid (Y_1,Z_1)$ which then implies $\Ic{U}{Y_1}{Z_2}\geq\Ic{U}{Z_1}{Z_2}$. Using the same argument it can be verified that also  $\Ic{U}{Y_2}{Y_1}\geq\Ic{U}{Z_2}{Y_1}$ holds, which proves inequality \eqref{eq:ineqCN}.

For a general $n \in \mathbb{N}$, using the same techniques as above, we obtain
\begin{align}
&\I{U}{Y^{n-1},Y_n}- \I{U}{Z^{n-1},Z_n} \nonumber \\
&\hspace{3mm}=\I{U}{Y^{n-1}}+\Ic{U}{Y_n}{Y^{n-1}}-\I{U}{Z_n}-\Ic{U}{Z^{n-1}}{Z_n}\\
&\hspace{3mm}=\Ic{U}{Y_n}{Y^{n-1}}-\Ic{U}{Z^{n-1}}{Z_n}+\I{U}{Y^{n+1}}-\I{U}{Z_n}\nonumber\\
&\hspace{8mm}+\I{U}{Y^{n-1},Z_n}-\I{U}{Y^{n-1},Z_n}\\
&\hspace{3mm}= \Ic{U}{Y^{n-1}}{Z_n}-\Ic{U}{Z^{n-1}}{Z_n}+\Ic{U}{Y_n}{Y^{n-1}}-\Ic{U}{Z_n}{Y^{n-1}}\\
&\hspace{3mm}\geq \Ic{U}{Y^{n-1}}{Z_n}-\Ic{U}{Z^{n-1}}{Z_n} \label{eq:ineq1}\\
&\hspace{3mm}= \Ic{U}{Y^{n-2}}{Z_n} + \Ic{U}{Y_{n-1}}{Y^{n-2},Z_n}-\Ic{U}{Z_{n-1}}{Z_n}-\Ic{U}{Z^{n-2}}{Z_{n-1},Z_n}\\
&\hspace{3mm}= \Ic{U}{Y_{n-1}}{Y^{n-2},Z_n}-\Ic{U}{Z^{n-2}}{Z_{n-1},Z_n}+\Ic{U}{Y^{n-2}}{Z_n}-\Ic{U}{Z_{n-1}}{Z_n}\nonumber\\
& \hspace{8mm}+\Ic{U}{Y^{n-2},Z_{n-1}}{Z_n}-\Ic{U}{Y^{n-2},Z_{n-1}}{Z_n}\\
&\hspace{3mm}= \Ic{U}{Y^{n-2}}{Z_{n-1},Z_n}-\Ic{U}{Z^{n-2}}{Z_{n-1},Z_n}+\Ic{U}{Y_{n-1}}{Y^{n-2},Z_n} \nonumber \\
& \hspace{8mm}-\Ic{U}{Z_{n-1}}{Y^{n-2},Z_n}\\
&\hspace{3mm}\geq \Ic{U}{Y^{n-2}}{Z_{n-1},Z_n}-\Ic{U}{Z^{n-2}}{Z_{n-1},Z_n}. \label{eq:ineq2}
\end{align}
The equality steps follow by applying the chain rule. Inequality~\eqref{eq:ineq1} is valid since by assumption we have $\W \succeq \V$ which implies that $\Ic{U}{Y_n}{Y^{n-1}}\geq\Ic{U}{Z_n}{Y^{n-1}}$ as explained above. Inequality~\eqref{eq:ineq2} follows by similar arguments.
Continuing applying these steps we obtain
\begin{align}
\I{U}{Y^{n-1},Y_n}- \I{U}{Z^{n-1},Z_n} \geq \Ic{U}{Y_1}{Z_2^n}-\Ic{U}{Z_1}{Z_2^n} \geq 0,
\end{align}
which proves the assertion.
\end{proof}

\section{Universality for More Capable Sources} \label{sec:universalityMC}
We consider $n$ i.i.d.\ copies of a triple of correlated random variables $(X,Y,Z)^n$ described by $(P_{X,Y,Z})^n$. Let $\cY:\mathcal{X}^n \to \{0,1\}^{k_Y}$ be an optimal compressor for $X^n$ given the side information $Y^n$. Being optimal implies that
\begin{enumerate}[(a)]
\item \label{item:a} $\Hcmin{X^n}{\cY(X^n),Y^n}\leq \epsilon_n$ for $\epsilon_n \geq 0$ such that $\lim \limits_{n \to \infty} \epsilon_n =0$, since we want to be able to recover $X^n$ out of $(\cY(X^n),Y^n)$ perfectly in the limit $n\to \infty$. This can be justified as follows by the following argument
\begin{align}
\Perr{X^n}{\cY(X^n),Y^n} &= 1-\Pguess{X^n}{\cY(X^n),Y^n}\\
 &= 1-2^{-\Hcmin{X^n}{\cY(X^n),Y^n}},
\end{align}
where the first equality is by definition and the second equality is due to \cite[Theorem 1]{koenig09}.
\item \label{item:b}$ \lim \limits_{n \to \infty} \frac{k_Y}{n}=\Hc{X}{Y}$, since the compressor should achieve the optimal rate, given by Slepian and Wolf \cite{slepian73}.
\end{enumerate}
Using polar codes a possible optimal compressor for $X^n$ given the side information $Y^n$ has the form $\cY(X^n)=U[\mathcal{D}_{\epsilon}^n(X|Y)^{\setC}]$, where $U^n=G_n X^n$ and $\mathcal{D}_{\epsilon}^n(X|Y)$ as defined in \eqref{eq:DsetY}. The proof that this compressor is optimal, i.e.,\ fulfills (\ref{item:a}) and (\ref{item:b}) has been given in \cite{arikan10}.

Let $\cZ$ be an optimal compressor for $X^n$ given the side information $Z^n$.
In the same spirit as \c Sa\c so\u glu did for the channel coding scenario (cf.\ Propoposition~\ref{prop:sasoglu}) one can prove the following proposition.
\begin{myprop}\label{prop:sasogluSC}
Let $P_{X,Y}$ and $P_{X,Z}$ be two sources such that $P_{Y|X}$ is more capable than $P_{Z|X}$, then any optimal compressor $\cZ$ for $P_{X,Z}$ can be also used for $P_{X,Y}$, or more precisely 
\begin{equation}
\Perr{X^n}{\cZ(X^n),Y^n} \leq n \, \Perr{X^n}{\cZ(X^n),Z^n} + \Hb(\Perr{X^n}{\cZ(X^n),Z^n}),
\end{equation}
under the use of an optimal decoder.
\end{myprop}
\begin{proof}
Using Fano's inequality we obtain
\begin{align}
n\, \Perr{X^n}{\cZ(X^n),Z^n} + \Hb(\Perr{X^n}{\cZ(X^n),Z^n}) &\geq \Hc{X^n}{\cZ(X^n),Z^n}\\
 &\geq \Hc{X^n}{\cZ(X^n),Y^n}\\
 & \geq 1-e^{-\Hc{X^n}{\cZ(X^n),Y^n}}\label{eq:easy}\\
 & \geq \Perr{X^n}{\cZ(X^n),Y^n}.
\end{align}
The second inequality follows by Proposition~\ref{prop:csiszarTrick} and the more capable relation that, according to Proposition~\ref{prop:tensorMC}, is preserved when considering $n$ i.i.d.\ copies of the source $P_{X,Y}$ respectively $P_{X,Z}$. Note that since $\cZ(X^n)$ is a deterministic function of $X^n$ the following Markov condition holds $\cZ(X^n) \markovDavid X^n \markovDavid (Y^n,Z^n)$. Inequality \eqref{eq:easy} uses the simple fact that $1-e^{-x}\leq x$ for $x\geq 0$. The final inequality can be justified as follows
\begin{align}
\Perr{X^n}{\cZ(X^n),Y^n} &= 1-\Pguess{X^n}{\cZ(X^n),Y^n}\\
 &= 1-2^{-\Hcmin{X^n}{\cZ(X^n),Y^n}}\\
 &\leq 1-2^{-\Hc{X^n}{\cZ(X^n),Y^n}},
\end{align}
where the first equality is by definition and the second equality is \cite[Theorem 1]{koenig09}.
\end{proof}

The following theorem is the main result of this section and implies that there exists a modified polar code, that inherits all the nice properties a standard polar code has, which is universal for a specific class of more capable sources. 

\begin{mythm} \label{thm:extension}
Consider two discrete memoryless sources $P_{X,Y}$ and $P_{X,Z}$ such that $P_{Y|X}\gg P_{Z|X}$ and $P_X$ is such that it maximizes $\I{X}{Y}-\I{X}{Z}$ for the given $P_{Y|X}$ and $P_{Z|X}$. Then for $\epsilon = O(2^{-n^{\beta}})$ with $\beta <\tfrac{1}{2}$, we have $\mathcal{D}_{\epsilon}^{n}(X|Y)^\setC$ $\scriptsize{\overset{\boldsymbol{\cdot}}{\subseteq}}$ $\mathcal{D}_{\epsilon}^{n}(X|Z)^\setC$.
\end{mythm}
\begin{proof}
A polar code for optimal Slepian-Wolf coding of $X^n$ given $Y^n$ defines a compressor $\cY:\mathcal{X}^n \to \{0,1\}^{k_Y}$, such that $\cY(X^n)=U^n[\mathcal{D}_{\epsilon}^n(X|Y)^{\setC}]$, where $U^n=G_n X^n$ and $\mathcal{D}_{\epsilon}^n(X|Y)$ is as defined in \eqref{eq:DsetY}. Since $\cY(X^N)$ is a deterministic function of $X^n$ we can write
\begin{align}
\lim \limits_{n\to \infty} \frac{1}{n} \Hc{X^n}{c_Y(X^n),Z^n} &= \lim \limits_{n \to \infty} \frac{1}{n} \left(\Hc{X^n}{Z^n} - \Hc{\cY(X^n)}{Z^n}\right) \label{eq:middle}\\
&= \Hc{X}{Z} - \lim \limits_{n \to \infty} \frac{1}{n} \Hc{\cY(X^n)}{Z^n}.\label{eq:finalstep}
\end{align}
The rate of the private channel coding protocol introduced in \cite[Theorem 7]{sutter13} is given as
\begin{equation}
R =\lim \limits_{n\to \infty} \frac{1}{n} \Hc{X^n}{\cY(X^n),Z^n}. \label{eq:rate}
\end{equation}
Using \eqref{eq:finalstep}, \eqref{eq:rate} and the cardinality bound for the entropy, we obtain a lower bound for the rate
\begin{equation}
R \geq \Hc{X}{Z}-\Hc{X}{Y}.
\end{equation}
However, since we are in a more capable scenario and we assume that $P_X$ is such that it maximizes $\Hc{X}{Z}-\Hc{X}{Y}$, it is known that this bound is equal to the secrecy capacity \cite[Page 551]{elgamal12} and hence $R=\Hc{X}{Z}-\Hc{X}{Y}$ holds, which implies that
\begin{align}
\lim \limits_{n \to \infty} \frac{1}{n} \Hc{\cY(X^n)}{Z^n} &= \Hc{X}{Y} \label{eq:1}\\
&= \lim \limits_{n \to \infty} \frac{1}{n} \left|\cY \right|. \label{eq:2}
\end{align}
This proves that $\cY$ extracts an essentially maximal entropy set also for side information $Z^n$. Recall that the optimal decompressors using polar codes have the form $\cY(X^n)=U^n[\mathcal{D}_{\epsilon}^n(X|Y)^{\setC}]$ and $\cZ(X^n)=U^n[\mathcal{D}_{\epsilon}^n(X|Z)^{\setC}]$ as explained above. Using \eqref{eq:1} and \eqref{eq:2} gives
\begin{align}
\lim \limits_{n \to \infty} \frac{1}{n} \Hc{U^n[\mathcal{D}_{\epsilon}^n(X|Y)^{\setC}]}{Z^n} &= \Hc{X}{Y}\\
&= \lim \limits_{n \to \infty} \frac{1}{n} \left| \mathcal{D}_{\epsilon}^n(X|Y)^{\setC} \right|,
\end{align}
which implies that $\mathcal{D}_{\epsilon}^n(X|Y)^{\setC}$ $\scriptsize{\overset{\boldsymbol{\cdot}}{\subseteq}}$ $\mathcal{D}_{\epsilon}^n(X|Z)^{\setC}$.
\end{proof}
Theorem~\ref{thm:extension} directly implies that for two discrete memoryless sources $P_{X,Y}$ and $P_{X,Z}$ such that $P_{Y|X}\gg P_{Z|X}$ and $P_X$ such that it maximizes $\I{X}{Y}-\I{X}{Z}$ for the given $P_{Y|X}$ and $P_{Z|X}$, for every polar code constructed for $P_{X,Z}$ there exists a slightly modified polar code---having the same rate, the same encoding and decoding complexity and the same error rate---that can be used reliably for $P_{X,Y}$ under SC decoding.  More precisely, given a polar code characterized by $\mathcal{D}_{\epsilon}^n(X|Z)^\setC$, we can define a modified polar code that is characterized by $\tilde{\mathcal{D}_{\epsilon}^n}(X|Z)^\setC :=\mathcal{D}_{\epsilon}^n(X|Z)^\setC \cup \mathcal{A}_{\epsilon}^n$, where $\mathcal{A}_{\epsilon}^n:=\mathcal{D}_{\epsilon}^n(X|Y)^\setC \backslash \mathcal{D}_{\epsilon}^n(X|Z)^\setC$. Note that these two codes have the same rate since $\lim_{n\to\infty} \tfrac{1}{n}| \mathcal{D}_{\epsilon}^n(X|Z)^\setC| = \lim_{n\to \infty} \tfrac{1}{n} | \tilde{\mathcal{D}_{\epsilon}^n}(X|Z)^\setC|=\Hc{X}{Z}$ as predicted by Theorem~\ref{thm:extension}. Moreover, we can use the same encoding and decoding for the modified polar code as for the standard polar code, such that it inherits its low complexity and error rate.
The modified polar code described by $ \tilde{\mathcal{D}_{\epsilon}^n}(X|Z)^\setC$ can be used reliably for $P_{X,Y}$ using SC decoding since $ \mathcal{D}_{\epsilon}^n(X|Y)^\setC\subseteq \tilde{\mathcal{D}_{\epsilon}^n}(X|Z)^\setC$.

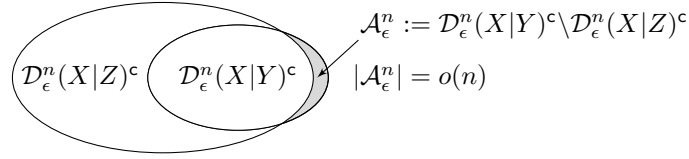
\begin{figure}[!htb]
\centering
\def \x{1.0}
\def \s{1.0}

\begin{tikzpicture}[scale=1,auto, node distance=1cm,>=latex']
	

 \draw[fill=gray!30] (0,0) ellipse (1.2cm and 0.7cm); 
\begin{scope}
\clip (0,0) ellipse (1.2cm and 0.7cm);
\clip (-\s,0) ellipse (2cm and 1.0cm);
\fill[white] (-2,1.5)
rectangle (2,-1.5);
\end{scope}
 \draw[] (-\s,0) ellipse (2cm and 1.0cm); 
  \draw[] (0,0) ellipse (1.2cm and 0.7cm); 

    \node at (0,0) {$\mathcal{D}_{\epsilon}^n(X|Y)^{\setC}$};    
    \node at (-2.1,0.0) {$\mathcal{D}_{\epsilon}^n(X|Z)^{\setC}$};       
 \draw[->] (1.6,0.5) -- (1.05,0);   
     \node at (3.8,0.7) {$\mathcal{A}_{\epsilon}^n:=\mathcal{D}_{\epsilon}^n(X|Y)^\setC \backslash \mathcal{D}_{\epsilon}^n(X|Z)^\setC$};  
     \node at (2.42,0.0) {$|\mathcal{A}_{\epsilon}^n|=o(n)$};  
\end{tikzpicture}
\caption{\small Construction of a modified polar code that is universal for two sources $P_{X,Y}$ and $P_{X,Z}$ where $P_{Y|X} \gg P_{Z|X}$ and $P_X$ is such that it maximizes $\I{X}{Y}-\I{X}{Z}$ for the given $P_{Y|X}$ and $P_{Z|X}$.}
\label{fig:SC}
\end{figure}

\section{Universality for More Capable Channels} \label{sec:universalityMCchannels}
Consider two DMCs $\W:\mathcal{X} \to \mathcal{Y}$ and $\V:\mathcal{X} \to \mathcal{Z}$ where $\W$ is assumed to be more capable than $\V$. In this section, we first recall a known result, that any optimal code built for $\V$ will perform also well when used over $\W$. Theorem~\ref{thm:extension}, then implies a stronger result for the use of polar codes and a particular input distribution, that for any polar code for $\V$ and a particular input distribution there exists a modified polar code---having the same rate, the same encoding and decoding complexity and the same error rate---that can be used reliably for $\W$ under SC decoding.
Using Fano's inequality \c Sa\c so\u glu proved the following result.
\begin{myprop}[\cite{sasoglu_phd}] \label{prop:sasoglu}
Consider two DMCs $\W$ and $\V$ and let $p_{e,\W}$ denote the average error probability of a code of length $n$ over a DMC $\W$, under optimal decoding. If $\W$ is more capable than $\V$, then
\begin{equation}
p_{e,\W} \leq n p_{e,\V} + \Hb(p_{e,\V}).
\end{equation}
\end{myprop}
\begin{proof}
The proof can be found in \cite[Chapter 7]{sasoglu_phd}.\footnote{\c Sa\c so\u glu formulated the statement for binary-input DMCs, however it can be verified that the proof remains valid for non-binary DMCs (having finite input alphabets).}
\end{proof}

\begin{mycor} \label{cor:extensionChannel}
Consider two DMCs $\W:\mathcal{X}\to \mathcal{Y}$ and $\V:\mathcal{X}\to \mathcal{Z}$ where $\W \gg \V$. Let the input distribution $P_X$ be such that it maximizes $\I{X}{Y}-\I{X}{Z}$. Then for $\epsilon = O(2^{-n^{\beta}})$ with $\beta < \tfrac{1}{2}$, we have $\mathcal{D}_{\epsilon}^n(X|Z)$ $\scriptsize{\overset{\boldsymbol{\cdot}}{\subseteq}}$ $\mathcal{D}_{\epsilon}^n(X|Y)$. 
\end{mycor}
\begin{proof}
This is an immediate corollary out of Theorem~\ref{thm:extension}.
\end{proof}
Corollary~\ref{cor:extensionChannel} implies that for any polar code for $\V$ with respect to an input distribution $P_X$ that maximizes $\I{X}{Y}-\I{X}{Z}$, there exists a modified polar code having the same properties (i.e.,\ the same rate, the same encoder and decoder, and the same error rate), that can be used reliably for $\W$ under SC decoding. More precisely, given a polar code for $\V$ characterized by $\mathcal{D}_{\epsilon}^n(X|Z)$ there exists a slightly modified polar code that is described by $\tilde{\mathcal{D}_{\epsilon}^n}(X|Z):=\mathcal{D}_{\epsilon}^n(X|Z)\backslash \mathcal{B}_{\epsilon}^n$, where $\mathcal{B}_{\epsilon}^n:=\mathcal{D}_{\epsilon}^n(X|Z) \backslash \mathcal{D}_{\epsilon}^n(X|Y)$, that has the same same rate since $\lim_{n \to \infty} \tfrac{1}{n}|\mathcal{D}_{\epsilon}^n(X|Y)|=\lim_{n \to \infty} \tfrac{1}{n}|\tilde{\mathcal{D}_{\epsilon}^n}(X|Y)|=1-\Hc{X}{Y}$ as predicted by Corollary~\ref{cor:extensionChannel}. In addition, the same encoder and decoder as for the standard polar code can be used also for the modified polar code such that it inherits the low complexity and error rate. The modified polar code described by $\tilde{\mathcal{D}_{\epsilon}^n}(X|Z)$ can be used reliably for $\W$ under SC decoding since $\tilde{\mathcal{D}_{\epsilon}^n}(X|Z) \subseteq \mathcal{D}_{\epsilon}^n(X|Y)$.

\begin{figure}[!htb]
\centering
\def \x{1.0}
\def \s{1.0}

\begin{tikzpicture}[scale=1,auto, node distance=1cm,>=latex']
	

 \draw[fill=gray!30] (0,0) ellipse (1.2cm and 0.7cm); 
\begin{scope}
\clip (0,0) ellipse (1.2cm and 0.7cm);
\clip (-\s,0) ellipse (2cm and 1.0cm);
\fill[white] (-2,1.5)
rectangle (2,-1.5);
\end{scope}
 \draw[] (-\s,0) ellipse (2cm and 1.0cm); 
  \draw[] (0,0) ellipse (1.2cm and 0.7cm); 

    \node at (0,0) {$\mathcal{D}_{\epsilon}^n(X|Z)$};    
    \node at (-2.1,0.0) {$\mathcal{D}_{\epsilon}^n(X|Y)$};       
 \draw[->] (1.6,0.5) -- (1.05,0);   
     \node at (3.8,0.7) {$\mathcal{B}_{\epsilon}^n:=\mathcal{D}_{\epsilon}^n(X|Z) \backslash \mathcal{D}_{\epsilon}^n(X|Y)$};  
     \node at (2.56,0.0) {$|\mathcal{B}_{\epsilon}^n|=o(n)$};  
\end{tikzpicture}
\caption{\small Construction of a modified polar code that is universal for two DMCs $\W:\mathcal{X}\to \mathcal{Y}$ and $\V:\mathcal{X}\to\mathcal{Z}$, where $\W \gg \V$ for an input distribution $P_X$ that it maximizes $\I{X}{Y}-\I{X}{Z}$.}
\label{fig:CC}
\end{figure}
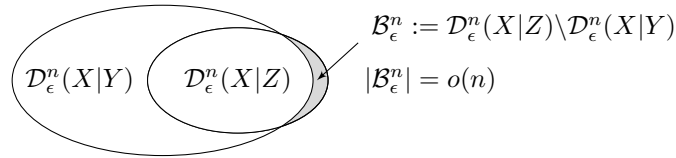

\section{Universality for less noisy channels} \label{sec:UniversalityPC}

Let $\W:\mathcal{X}\to \mathcal{Y}$ and $\V:\mathcal{X}\to \mathcal{Z}$ be two DMCs such that $\V$ is a degraded version of $\W$. It is known that the polar code constructed for $\V$ can be used reliably for communication over $\W$ using SC decoding \cite{arikan09,hassani09}. We next generalize this observation by showing that a less noisy relation is already sufficient to make this statement. Note that we assume that the decoder has full knowledge about the precise channel statistics.
\begin{mythm} \label{thm:LNpolar}
Let $\W:\mathcal{X} \to \mathcal{Y}$ and $\V:\mathcal{X}\to \mathcal{Z}$ be two DMCs such that $\W \succeq \V$, then for any $\epsilon \in (0,1)$, $n=2^k$, $k\in \mathbb{N}$ we have $\mathcal{D}_{\epsilon}^n(X|Z) \subseteq \mathcal{D}_{\epsilon}^n(X|Y)$.
\end{mythm}
\begin{proof}
Recall that the sets $\mathcal{D}_{\epsilon}^n(X|Y)$ and $\mathcal{D}_{\epsilon}^n(X|Z)$ are defined in \eqref{eq:DsetY} and \eqref{eq:DsetZ}. Let $U^n=G_n X^n$, where $G_n$ denotes the polarization transform and $P_{X^n}=(P_X)^n$. Using Proposition~\ref{prop:tensorEssLN} we have $\Hc{U_1}{Y^n} \leq \Hc{U_1}{Z^n}$ by definition of $\W^n \succeq \V^n$.
It thus remains to prove that $\Hc{U_i}{Y^n,U^{i-1}} \leq \Hc{U_i}{Z^n,U^{i-1}}$ for all $2\leq i\leq n$. Consider the following Markov chain $U^{i-1}\markovDavid U^i\markovDavid X^n \markovDavid (Y^n,Z^n)$ for the $U^{i-1}$ and $U^i$ as introduced above.\footnote{Note that it is easy to verify that the Markov condition is satisfied.} We then have
\begin{align}
\Hc{U_i}{Y^n,U^{i-1}} &=\Hc{U^{i}}{Y^n,U^{i-1}}\\
&\leq \Hc{U^{i}}{Z^n,U^{i-1}}\\
& = \Hc{U_i}{Z^n,U^{i-1}},
\end{align}
where the inequality step uses Proposition~\ref{prop:csiszarTrick} together with the fact that $\W^n \succeq \V^n$ which is ensured by Proposition~\ref{prop:tensorEssLN} and the assumption that $\W \succeq \V$.
\end{proof}

Having Theorem~\ref{thm:LNpolar} the main result of this section follows straightforwardly.
\begin{mycor} \label{cor:universalPC}
Let $\W:\mathcal{X}\to \mathcal{Y}$ and $\V:\mathcal{X} \to \mathcal{Z}$ be two DMCs such that $\W \succeq \V$, then $\mathcal{D}^{n}_{\epsilon}(X|Z) \subseteq \mathcal{D}^{n}_{\epsilon}(X|Y)$, i.e., the polar code constructed for $\V$ can be used for $\W$ under successive cancellation decoding. 
\end{mycor}

Note that the statement of Corollary~\ref{cor:universalPC} would be false if the assumption were relaxed to $\W \gg \V$, since \cite{hassani09} shows that, using polar codes with SC decoding, the compound capacity for $\BEC_C$ and $\BSC_C$ with $C=\tfrac{1}{2}$ is strictly smaller than $\frac{1}{2}$.

\section{Discussion}

For two arbitrary DMCs $\W:\mathcal{X}\to \mathcal{Y}$ and $\V:\mathcal{X}\to \mathcal{Z}$ it is of general interest to better understand the structure of the sets $\mathcal{D}_{\epsilon}^n(X|Y)$ and $\mathcal{D}_{\epsilon}^n(X|Z)$ as defined in \eqref{eq:DsetY} and \eqref{eq:DsetZ}. We have taken  a first step in this direction by proving that if $\W$ is more capable than $\V$ and considering an input distribution $P_X$ that maximizes $\I{X}{Y}-\I{X}{Z}$, we have $\mathcal{D}_{\epsilon}^n(X|Z)$ $\scriptsize{\overset{\boldsymbol{\cdot}}{\subseteq}}$ $\mathcal{D}_{\epsilon}^n(X|Y)$. For the stronger assumption that $\W$ is less noisy than $\V$ we showed that $\mathcal{D}_{\epsilon}^n(X|Z) \subseteq \mathcal{D}_{\epsilon}^n(X|Y)$ is valid for an arbitrary input distribution $P_X$, which was thus far only known for the setup where $\V$ is a degraded version of $\W$.

We would like to emphasize that the requirement of two DMCs being less noisy is in general much weaker than being degraded. Assume that we are given a polar code for a $\BSC(\alpha)$, $\alpha \in (0,\tfrac{1}{2})$. We would like to use this code for a $\BEC(\beta)$. For which values of $\beta$ is the code provably reliable using SC decoding? It is well known (cf.\ \cite[Example 5.4, Page 121]{elgamal12}) that for $0< \beta \leq 2 \alpha$, $\BSC(\alpha)$ is a degraded version of $\BEC(\beta)$. For $2\alpha <\beta \leq 4 \alpha(1-\alpha)$ we have $\BEC(\beta) \gg \BSC(\alpha)$ but $\BSC(\alpha)$ is not a degraded version of $\BEC(\beta)$. Thus according to Corollary~\ref{cor:universalPC}, we can choose $0<\beta \leq4 \alpha(1-\alpha)$, whereas before this new result only $0<\beta \leq 2  \alpha$ was proven to be reliable.

The result obtained in Theorem~\ref{thm:extension} might be also helpful for the code construction of polar codes. Let $\mathcal{C}_C \ni \W:\mathcal{X}\to \mathcal{Y}$ and $\BEC_C=\V:\mathcal{X}\to \mathcal{Z}$, then by Proposition~\ref{prop:extremes} we have $\V \gg \W$. In addition, let $P_X$ be such that it maximizes $\I{X}{Y}-\I{X}{Z}$. Recall that it is in general hard to compute $\mathcal{D}_{\epsilon}^n(X|Y)$, however $\mathcal{D}_{\epsilon}^n(X|Z)$ can be computed easily  \cite{arikan09,sasogluPC}. According to Corollary~\ref{cor:extensionChannel}, we have $\mathcal{D}_{\epsilon}^n(X|Y)$ $\scriptsize{\overset{\boldsymbol{\cdot}}{\subseteq}}$  $\mathcal{D}_{\epsilon}^n(X|Z)$. Hence, $\mathcal{D}_{\epsilon}^n(X|Y)$ is essentially contained in $\mathcal{D}_{\epsilon}^n(X|Z)$. It might be that the reduction from $\mathcal{D}_{\epsilon}^n(X|Z)$ to $\mathcal{D}_{\epsilon}^n(X|Y)$ is easier than to compute $\mathcal{D}_{\epsilon}^n(X|Y)$ by itself.

A better understanding of the structure of the high-entropy sets might also help to prove an open conjecture that states that the quantum polar codes introduced in \cite{renes12} do not need any entanglement assistance.

\vspace{6mm}

\noindent{\bf{Acknowledgments:}}
The authors would like to thank Chandra Nair for contributing the proof of Proposition~\ref{prop:tensorEssLN}. They further wish to thank Renato Renner, Omar Fawzi, Jossy Sayir and Seyed Hamed Hassani for helpful discussions. This work was supported by the Swiss National Science Foundation (through the National Centre of Competence in Research `Quantum Science and Technology' and grant No.~200020-135048) and by the European Research Council (grant No.~258932).




\bibliography{./bibtex/header,./bibtex/bibliofile}


\clearpage

\end{document}